\documentclass[a4paper, 11pt]{article}

\usepackage{fullpage}

\usepackage{amsfonts}
\usepackage{amssymb}
\usepackage{amsmath}
\usepackage{amsthm}
\usepackage{latexsym}
\usepackage{url}

\usepackage{hyperref}

\newcommand{\norm}[1]{\| #1 \|}

\newcommand{\deltamaps}{\overset{\Delta_1}\longmapsto}
\newcommand{\group}{\mathbb{G}}
\newcommand{\Adv}{\mathrm{Adv}^\pm}

\newcommand{\domain}{\mathcal{\widetilde D}}

\newtheorem{theorem}{Theorem}
\newtheorem{lemma}[theorem]{Lemma}
\theoremstyle{remark}
\newtheorem{definition}[theorem]{Definition}
\newtheorem{example}[theorem]{Example}

\newcommand{\refsec}[1]{Section~\ref{sec:#1}}

\newcommand{\refeqn}[1]{(\ref{eqn:#1})}

\newcommand{\refthm}[1]{Theorem~\ref{thm:#1}}
\newcommand{\reflem}[1]{Lemma~\ref{lem:#1}}

\newcommand{\cD}{{\cal D}}
\newcommand{\tGamma}{{\widetilde{\Gamma}}}
\newcommand{\eps}{{\varepsilon}}
\newcommand{\tG}{{\widetilde{G}}}
\newcommand{\pfstart}{\begin{proof}} 
\newcommand{\pfend}{\end{proof}}

\newcommand{\ignore}[1]{}
\newcommand{\elem}[1]{[\![#1]\!]}

\begin{document}
\title{Adversary Lower Bound for the $k$-sum Problem}
\author{Aleksandrs Belovs%
\thanks{Faculty of Computing, University of Latvia}\\
{\tt stiboh@gmail.com}
 \and
Robert \v Spalek%
  \thanks{%
  Google, Inc.}\\
  {\tt spalek@google.com}
}
\date{}
\maketitle

\begin{abstract}
We prove a tight quantum query lower bound $\Omega(n^{k/(k+1)})$ for the problem of deciding whether there exist $k$ numbers among $n$ that sum up to a prescribed number, provided that the alphabet size is sufficiently large.

This is an extended and simplified version of an earlier preprint of one of the authors~\cite{belovs:adv-el-dist}.
\end{abstract}


\section{Introduction}
Two main techniques for proving lower bounds on quantum query complexity are the polynomial method~\cite{beals:polynomial} developed by Beals {\em et al.} in 1998, and the adversary method~\cite{ambainis:adversary} developed by Ambainis in 2000. Both techniques are incomparable. There are functions with adversary bound strictly larger than polynomial degree~\cite{ambainis:polynomialVsQCC}, as well as functions with the reverse relation.

One of the examples of the reverse relation is exhibited by the element distinctness function. The input to the function is a string of length $n$ of symbols in an alphabet of size $q$, i.e., $x = (x_i)\in [q]^n$. We use notation $[q]$ to denote the set $\{1,\dots,q\}$. The element distinctness function evaluates to 0 if all symbols in the input string are pairwise distinct, and to 1 otherwise.

The quantum query complexity of element distinctness is $O(n^{2/3})$ with the algorithm given by Ambainis~\cite{ambainis:distinctness}. The tight lower bounds were given by Aaronson and Shi~\cite{shi:collisionLower}, Kutin~\cite{kutin:collisionLower} and Ambainis~\cite{ambainis:collisionLower} using the polynomial method. 

The adversary bound, however, fails for this function. The reason is that the function has 1-certificate complexity 2, and the so-called certificate complexity barrier~\cite{spalek:adversaryEquivalent, zhang:adversaryPower} implies that for any function with 1-certificate complexity bounded by a constant, the adversary method fails to achieve anything better than $\Omega(\sqrt{n})$.

In 2006 a stronger version of the adversary bound was developed by H\o yer {\em et al.}~\cite{hoyer:adversaryNegative}. This is the negative-weight adversary lower bound defined in~\refsec{def}. Later it was proved to be optimal by Reichardt {\em et al.}~\cite{reichardt:adversaryTight, lee:stateConversion}. Although the negative-weight adversary lower bound is known to be tight, it has almost never been used to prove lower bounds for explicit functions. Vast majority of lower bounds by the adversary method used the old positive-weight version of this method. But since the only competing polynomial method is known to be non-tight, a better understanding of the negative-weight adversary method would be very beneficial. In the sequel, we consider the negative-weight adversary bound only, and we will omit the adjective ``negative-weight''.

In this paper we use the adversary method to prove a lower bound for the following variant of the knapsack packing problem. Let $\group$ be a finite Abelian group, and $t\in \group$ be its arbitrary element. For a positive integer $k$, the {\em k-sum problem} consists in deciding whether the input string $x_1,\dots,x_n\in \group$ contains a subset of $k$ elements that sums up to $t$. We assume that $k$ is an arbitrary but fixed constant. The main result of the paper is the following
\begin{theorem}
\label{thm:ksum}
For a fixed $k$, the quantum query complexity of the $k$-sum problem is $\Omega(n^{k/(k+1)})$ provided that $|\group|\ge n^k$.
\end{theorem}
Clearly, the 1-certificate complexity of the $k$-sum problem is $k$, hence, it is also subject to the certificate complexity barrier. 

The result of \refthm{ksum} is tight thanks to the quantum algorithm based on quantum walks on the Johnson graph \cite{ambainis:distinctness}. This algorithm was first designed to solve the $k$-distinctness problem. This problem asks for detecting whether the input string $x\in [q]^n$ contains $k$ elements that are all equal. Soon enough it was realized that the same algorithm works for any function with 1-certificate complexity $k$~\cite{childs:subsetFinding}, in particular, for the $k$-sum problem. The question of the tightness of this algorithm remained open for a long time. It was known to be tight for $k=2$ due to the lower bound for the element distinctness problem. Now we know that it is not optimal for the $k$-distinctness problem if $k>2$~\cite{belovs:learningKDist}. However, \refthm{ksum} shows that, for every $k$, quantum walk on the Johnson graph is optimal for some functions with 1-certificate complexity $k$.  Finally, we note that the $k$-sum problem is also interesting because of its applications in quantum Merkle puzzles~\cite{brassard:merkle, kalach:personal}. 

Actually, we get \refthm{ksum} as a special case of a more general result we are about to describe. The following is a special case of a well-studied combinatorial object:
\begin{definition}\label{def:orthogonalArray}
Assume $T$ is a subset of $[q]^k$ of size $q^{k-1}$.  We say that $T$ is an \emph{orthogonal array of length $k$} iff, for every index $i \in [k]$ and for every vector $x_1,\dots, x_{i-1}, x_{i+1},\dots, x_k \in [q]$, there exists exactly one $x_i \in [q]$ such that $(x_1, \dots, x_k) \in T$.
\end{definition}

For $x=(x_i)\in [q]^n$ and $S\subseteq [n]$ let $x_S$ denote the projection of $x$ on $S$, i.e., the vector $(x_{s_1},\dots,x_{s_\ell})$ where $s_1,\dots,s_\ell$ are the elements of $S$ in the increasing order.

Assume each subset $S$ of $[n]$ of size $k$ is equipped with an orthogonal array $T_S$. The {\em $k$-orthogonal array} problem consists in finding an element of any of the orthogonal arrays in the input string. More precisely, the input $x\in [q]^n$ evaluates to 1 iff there exists $S\subseteq [n]$ of size $k$ such that $x_S \in T_S$.

Consider the following two examples:

\begin{example}
Let $\group$ be a commutative group with $q$ elements and $t\in\group$.  $T = \{x \in \group^k: \sum_{i=1}^k x_i = t\}$ is an orthogonal array of length $k$.  This choice corresponds to the $k$-sum problem of \refthm{ksum}.
\end{example}

\begin{example}
\label{exm:distinctness}
$T = \{x \in [q]^2: x_1=x_2\}$ is an orthogonal array of length 2.  This corresponds to the element distinctness problem from~\cite{belovs:adv-el-dist}.
\end{example}

\goodbreak
\begin{theorem}
\label{thm:orthogonal}
For a fixed $k$ and any choice of the orthogonal arrays $T_S$, the quantum query complexity of the $k$-orthogonal array problem is $\Omega(n^{k/(k+1)})$ provided that $q \ge n^k$. The constant behind big-Omega depends on $k$, but not on $n$, $q$, or the choice of $T_S$.
\end{theorem}

The orthogonal array condition specifies that even if an algorithm has queried $k-1$ elements out of any $k$-tuple, it has the same information whether this $k$-tuple is a 1-certificate as if it has queried no elements out of it. Because of this, the search for a $k$-tuple as a whole entity is the best the quantum algorithm can do. Our proof of \refthm{orthogonal} is a formalization of this intuition.

Let us elaborate on the requirement on the size of the alphabet. It is easy to see that some requirement is necessary. Indeed, the $k$-sum problem can be solved in $O(\sqrt{n})$ queries if the size of $\group$ is $O(1)$, using the Grover search to find up to $k$ copies of every element of $\group$ in the input string, and trying to construct $t$ out of what is found. In some cases, e.g., when $t$ is the identity element and $k$ equals the order of the group, the problem becomes trivial if $n$ is large enough.

The requirement on the size of the alphabet for the element distinctness problem is a subtle issue. The lower bounds by Aaronson and Shi~\cite{shi:collisionLower} and Kutin~\cite{kutin:collisionLower} require the size of the alphabet to be at least $\Omega(n^2)$ that is the same that gives \refthm{orthogonal}. However, later Ambainis~\cite{ambainis:collisionLower} showed that the lower bound remains the same even if one allows the alphabet of size $n$. Reducing the alphabet size in \refthm{orthogonal} is one of our open problems.

\section{Adversary Lower Bound}
\label{sec:def}
In the paper we are interested in the quantum query complexity of solving the orthogonal array problem.  For the definitions and main properties of quantum query complexity refer to, e.g., Ref.~\cite{buhrman:querySurvey}.  For the purposes of our paper, it is enough with the definition of the adversary bound we give in this section.

Compared to the original formulation of the negative adversary bound~\cite{hoyer:adversaryNegative}, our formulation has two differences.  Firstly, in order to simplify notations we call an adversary matrix a matrix with rows labeled by positive inputs, and the columns by the negative ones. It is a quarter of the original adversary matrix that completely specifies the latter. Secondly, due to technical reasons, we allow several rows to be labeled by the same positive input. All this is captured by the following definition and theorem.


\begin{definition}
\label{defn:adversary}
Let $f$ be a function $f\colon \cD\to \{0,1\}$ with domain $\cD\subseteq [q]^n$.  Let $\domain$ be a set of pairs $(x,a)$ with the property that the first element of each pair belongs to $\cD$, and $\domain_i = \{(x,a)\in \domain : f(x)=i\}$ for $i\in\{0,1\}$.  An {\em adversary matrix} for the function $f$ is a non-zero real $\domain_1\times\domain_0$ matrix $\Gamma$. And, for $i\in[n]$, let $\Delta_i$ denote the $\domain_1\times \domain_0$ matrix defined by
\[ \Delta_i\elem{(x,a),(y,b)} = \begin{cases} 0,& x_i=y_i; \\ 1,&\text{otherwise}. \end{cases} \]
\end{definition}

\begin{theorem}[Adversary bound \cite{hoyer:adversaryNegative}] \label{thm:adv}
In the notation of Definition~\ref{defn:adversary}, $Q_2(f)=\Omega(\Adv(f))$, where
\begin{equation}\label{eqn:adversary}
\Adv(f) = \sup_{\Gamma} \frac{\norm{\Gamma}}{\max_{i\in n} \norm{\Gamma\circ\Delta_i} }
\end{equation}
with the maximization over all adversary matrices for $f$, $\norm{\cdot}$ is the spectral norm, and $Q_2(f)$ is the quantum query complexity of $f$.
\end{theorem}

\begin{proof}
In the original negative-weight adversary bound paper~\cite{hoyer:adversaryNegative}, Eq.~\refeqn{adversary} is proven when $\Gamma$ is a real symmetric $\cD\times \cD$ matrix with the property $\Gamma\elem{x,y}=0$ if $f(x)=f(y)$, and $\Delta_i$ are modified accordingly.  We describe a reduction from the adversary matrix in our definition, $\Gamma$, to the adversary matrix $\Gamma'$ in the definition of~\cite{hoyer:adversaryNegative}. Also, let $\Delta'_i$ be the $\cD\times\cD$ matrix with $\Delta'_i\elem{x,y}=1$ if $x_i\ne y_i$, and 0, otherwise.

At first, define $\overline\Gamma$ as 
\[
\overline\Gamma = \left( \begin{matrix}
0 & \Gamma^* \\
\Gamma & 0 \\
\end{matrix} \right) \enspace.
\]
Note that $\norm{\overline\Gamma} = \norm{\Gamma}$, and the spectrum of $\overline\Gamma$ is symmetric.  Also, for all $i$, $\norm{\overline\Gamma\circ\overline\Delta_i} = \norm{\Gamma\circ\Delta_i}$, where $\overline\Delta_i$ is defined similarly to $\overline\Gamma$.

Let $\delta = (\delta_{x,a})$ be the normalized eigenvalue $\norm{\Gamma}$ eigenvector of $\overline\Gamma$.  For all $x,y\in \cD$, let:
\[\delta'_x = \sqrt{\sum_{a:(x,a)\in\widetilde{\cD}} \delta_{x,a}^2}\qquad\mbox{and}\qquad
 \Gamma'\elem{x,y} = \frac1{\delta'_x\delta'_y}\sum_{\substack{a:(x,a)\in\widetilde{\cD}\\ b:(y,b)\in \widetilde{\cD}}} \delta_{x,a}\delta_{y,b} \overline\Gamma\elem{(x,a),(y,b)} \enspace. \]
Then it is easy to see that $\delta' = (\delta'_{x})$ satisfies $\|\delta'\|=1$ and $(\delta')^* \Gamma'\delta' = \delta^*\overline \Gamma \delta = \norm{\Gamma}$, hence, $\|\Gamma'\|\ge \|\Gamma\|$.

And vice versa, if $\eps'$ is such that $\|\eps'\|=1$ and $(\eps')^* (\Gamma'\circ \Delta'_i)\eps' = \|\Gamma'\circ \Delta'_i\|$, let $\eps_{x,a} = \delta_{x,a} \eps'_x/\delta'_x$. Again, $\|\eps\|=1$ and $\eps^*(\overline\Gamma\circ\overline\Delta_i)\eps = (\eps')^*(\Gamma'\circ\Delta'_i)\eps'$, hence, $\|\Gamma'\circ \Delta'_i\| \le \|\overline\Gamma\circ \overline \Delta_i\| = \norm{\Gamma\circ\Delta_i}$.

This means that $\Gamma'$ provides at least as good adversary lower bound as $\Gamma$. 
\end{proof}


\section{Proof}
In this section we prove \refthm{orthogonal} using the adversary lower bound, \refthm{adv}.  The idea of our construction is to embed the adversary matrix $\Gamma$ into a slightly larger matrix $\tGamma$ with additional columns. Then $\Gamma\circ \Delta_i$ is a submatrix of $\tGamma\circ \Delta_i$, hence, $\|\Gamma\circ \Delta_i\|\le \|\tGamma\circ \Delta_i\|$.  (In this section we use $\Delta_i$ to denote all matrices defined like in Definition~\ref{defn:adversary}, with the size and the labels of the rows and columns clear from the context.)  It remains to prove that $\|\tGamma\|$ is large, and that $\|\Gamma\|$ is not much smaller than $\|\tGamma\|$.

The proof is organized as follows.  In \refsec{tGamma} we define $\tGamma$ in dependence on parameters $\alpha_m$, in \refsec{normtGamma} we analyze its norm, in Sections~\ref{sec:Delta1} and~\ref{sec:tGamma1Norm} we calculate $\norm{\tGamma\circ\Delta_i}$, in \refsec{alpha} we optimize $\alpha_m$s, and, finally, in \refsec{Gamma} we prove that the norm of the true adversary matrix $\Gamma$ is not much smaller than the norm of $\tGamma$.

\subsection{Adversary matrix}
\label{sec:tGamma}
Matrix $\tGamma$ consists of $\binom n k$ matrices $\tG_{s_1,\dots,s_k}$ stacked one on another for all possible choices of $S=\{s_1,\dots,s_k\}\subset[n]$:
\begin{equation}\label{eqn:tGamma}
\tGamma = \left(
\begin{array}{c} \tG_{1,2,\dots,k} \\ \tG_{1,2,\dots,k-1,k+1} \\ \dots \\ \tG_{n-k+1,n-k+2,\dots,n} \\ \end{array}
\right) \enspace.
\end{equation}
Each $\tG_S$ is a $q^{n-1} \times q^n$ matrix with rows indexed by inputs $(x_1, \dots, x_n) \in [q]^n$ such that $x_S \in T_S$, and columns indexed by all possible inputs $(y_1, \dots, y_n) \in [q]^n$.

We say a column with index $y$ is {\em illegal } if $y_S\in T_S$ for some $S\subseteq[n]$. After removing all illegal columns, $\tG_S$ will represent the part of $\Gamma$ with the rows indexed by the inputs having an element of the orthogonal array on $S$.  Note that some positive inputs appear more than once in $\Gamma$. More specifically, an input $x$ appears as many times as many elements of the orthogonal arrays it contains. 

This construction may seem faulty, because there are elements of $[q]^n$ that are used as labels of both rows and columns in $\tGamma$, and hence, it is trivial to construct a matrix $\tGamma$ such that the value in~\refeqn{adversary} is arbitrarily large. But we design $\tGamma$ in a specifically restrictive way so that it still is a good adversary matrix after the illegal columns are removed. 

Let $J_q$ be the $q\times q$ all-ones matrix. Assume $e_0,\dots,e_{q-1}$ is an orthonormal eigenbasis of $J_q$ with $e_0=1/\sqrt{q}(1,\dots,1)$ being the eigenvalue $q$ eigenvector. Consider the vectors of the following form:
\begin{equation}
\label{eqn:v}
v = e_{v_1}\otimes e_{v_2}\otimes\cdots\otimes e_{v_n} \enspace,
\end{equation}
where $v_i \in \{0,\dots,q-1\}$. These are eigenvectors of the Hamming Association Scheme on $[q]^n$. For a vector $v$ from~\refeqn{v}, the {\em weight} $|v|$ is defined as the number of non-zero entries in $(v_1,\dots,v_n)$. Let $E_k^{(n)}$, for $k=0,\dots,n$, be the orthogonal projector on the space spanned by the vectors from~\refeqn{v} having weight $k$. These are the projectors on the eigenspaces of the association scheme. Let us denote $E_i = E^{(1)}_i$ for $i=0,1$. These are $q\times q$ matrices. All entries of $E_0$ are equal to $1/q$, and the entries of $E_1$ are given by
\[
E_1\elem{x,y} = \begin{cases}
1-1/q,& x=y;\\
-1/q,& x\ne y.
\end{cases}
\]

Elements of $S$ in $\tG_S$ should be treated differently from the rest of the elements. For them, we define a $q^{k-1}\times q^k$ matrix $F_S$. It has rows labeled by the elements of $T_S$ and columns by the elements of $[q]^k$, and is defined as follows.

\begin{definition}\label{def:F}
Let 
\[E^{(k)}_{<k} = I - E^{(k)}_k = \sum_{i=0}^{k-1} E^{(k)}_i = \sum_{\substack{u = e_{u_1} \otimes \cdots \otimes e_{u_k}\\ |u|<k}} uu^*\]
 be the projector onto the subspace spanned by the vectors of less than maximal weight.  Let $F_S$ be $\sqrt q$ times the sub-matrix of $E^{(k)}_{<k}$ consisting of only the rows from $T_S$. 
\end{definition}


Finally, we define $\tGamma$ as in~\refeqn{tGamma} with $\tG_S$ defined by
\begin{equation}
\label{eqn:GT}
\tG_S = \sum_{m=0}^{n-k} \alpha_m F_S \otimes E^{(n-k)}_m \enspace,
\end{equation}
where $F_S$ acts on the elements in $S$ and $E_m$ acts on the remaining $n-k$ elements. Coefficients $\alpha_m$ will be specified later.

\subsection{Norm of $\tGamma$}
\label{sec:normtGamma}

\begin{lemma}\label{lem:normGamma}
Let $\tGamma$ be like in~\refeqn{tGamma} with $\tG_S$ defined like in~\refeqn{GT}. Then
\begin{itemize}
\item[(a)] $\|\tGamma\| = \Omega(\alpha_0 n^{k/2})$,
\item[(b)] $\|\tGamma\| = O(\max_m \alpha_m n^{k/2})$.
\end{itemize}
\end{lemma}

\pfstart
Fix a subset $S$ and denote $T=T_S$ and $F=F_S$. Recall that $E^{(k)}_{<k}$ is the sum of $uu^*$ over all $u=e_{u_1}\otimes\cdots\otimes e_{u_k}$ with at least one $u_j$ equal to 0, and $F$ is the restriction of $E^{(k)}_{<k}$ to the rows in $T$.

For $u=e_{u_1}\otimes\cdots\otimes e_{u_k}$ and $\ell$ such that $u_\ell = 0$, let $u^{(\ell)}$ denote the $\sqrt{q}$ multiple of $u$ restricted to the elements in $T$. The reason for the superscript is that we consider the following process of obtaining $u^{(\ell)}$: we treat $T$ as $[q]^{k-1}$ by erasing the $\ell$-th element in any string of $T$, then $u^{(\ell)}$ coincides on this set with $u$ with the $\ell$-th term removed. 

In this notation, the contribution from $uu^*$ to $F$ equals $u^{(\ell_u)}u^*$, where $\ell_u$ is any position in $u$ containing $e_0$. In general, we do not know how $u^{(\ell)}$s relate for different $\ell$. However, we know that, for a fixed $\ell$, they are all orthogonal; and for any $\ell$, $(e_0^{\otimes k})^{(\ell)}$ is the vector $1/\sqrt{q^{k-1}}(1,\dots,1)$.

Let us start with proving (a). We estimate $\|\tGamma\|$ from below as $w^*\tGamma w'$, where $w$ and $w'$ are unit vectors with all elements equal. In other words, $\|\tGamma\|$ is at least the sum of all its entries divided by $\sqrt{\binom n k q^{2n-1}}$. In order to estimate the sum of the entries of $\tGamma$, we rewrite~\refeqn{GT} as
\begin{equation}
\label{eqn:GT2}
\tG_S = \alpha_0 e_0^{\otimes(n-1)}(e_0^{\otimes n})^* + {\sum_{u,v}} \alpha_{|v|} (u^{(\ell_u)}\otimes v)(u\otimes v)^* \enspace,
\end{equation}
where the summation is over all $u$ and $v$ such that at least one of them contains an element different from $e_0$. The sum of all entries in the first term of~\refeqn{GT2} is $\alpha_0 q^{n-1/2}$. The sum of each column in each of  $(u^{(\ell_u)}\otimes v)(u\otimes v)^*$ is zero because at least one of $u^{(\ell_u)}$ or $v$ sums up to zero. By summing over all $\binom n k$ choices of $S$, we get that $\norm{\tGamma} \ge \alpha_0\sqrt{\binom n k} = \Omega(\alpha_0 n^{k/2})$.

In order to prove (b), express $F_S$ as $\sum_{\ell=1}^k F^{(\ell)}_S$ with $F^{(\ell)}_S = \sum_{u\in U_\ell} u^{(\ell)}u^*$. Here $\{U_\ell\}$ is an arbitrary decomposition of all $u$ such that $U_\ell$ contains only $u$ with $e_0$ in the $\ell$-th position. Define $\tG^{(\ell)}_S$ as in~\refeqn{GT} with $F_S$ replaced by $F^{(\ell)}_S$, and $\tGamma^{(\ell)}$ as in~\refeqn{tGamma} with $\tG_S$ replaced by $\tG^{(\ell)}_S$.

Since all $u^{(\ell)}$s are orthogonal for a fixed $\ell$, we get that
\[
(\tG^{(\ell)})^*\tG^{(\ell)} = \sum_{u\in U_\ell, v} \alpha_{|v|}^2 (u\otimes v)(u\otimes v)^* \enspace,
\]
thus $\| (\tG^{(\ell)})^*\tG^{(\ell)} \| = \max_m \alpha_m^2$. By the triangle inequality, 
\[
\|\tGamma^{(\ell)}\|^2 = \left\|(\tGamma^{(\ell)})^*\tGamma^{\ell}\right\| =  \left\|\sum_S (\tG_S^{(\ell)})^*\tG_S^{(\ell)} \right\| \le 
\binom n k \max_m \alpha_m^2 \enspace. 
\]
Since $\tGamma = \sum_{\ell=1}^k \tGamma^{(\ell)}$, another application of the triangle inequality finishes the proof of (b).
\pfend

\subsection{Action of $\Delta_1$}
\label{sec:Delta1}

The adversary matrix is symmetric in all input variables and hence it suffices to only consider the entry-wise multiplication by $\Delta_1$.  Precise calculation of $\|\tGamma \circ \Delta_1\|$ is very tedious, but one can get an asymptotically tight bound using the following trick.
Instead of computing $\tGamma \circ \Delta_1$ directly, we arbitrarily map $\tGamma \deltamaps \tGamma_1$ such that $\tGamma_1 \circ \Delta_1 = \tGamma \circ \Delta_1$, and use the inequality $\|\tGamma_1 \circ \Delta_1\| \le 2 \|\tGamma_1\|$ that holds thanks to $\gamma_2(\Delta_1)\le 2$ \cite{lee:stateConversion}.  In other words, we change arbitrarily the entries with $x_1 = y_1$.  We use the mapping
\begin{equation}\label{eqn:E1}
E_0 \deltamaps E_0 \enspace, \qquad E_1 \deltamaps -E_0 \enspace, \qquad I \deltamaps 0 \enspace.
\end{equation}
The projector $E^{(k)}_{<k}$ is mapped by $\Delta_1$ as
\[
E^{(k)}_{<k} = I - E^{(k)}_k \deltamaps E_0 \otimes E_1^{\otimes (k-1)} \enspace.
\]
It follows that
\begin{equation}\label{eqn:F1}
F \deltamaps e_0^*\otimes E_1^{\otimes (k-1)} = \sum_{\substack{u=e_{u_1}\otimes\cdots\otimes e_{u_k}\\ u_0=0, |u|=k-1}} u^{(1)} u^* \enspace,
\end{equation}
where $u^{(1)}$ is defined like in the proof of \reflem{normGamma}. 

\subsection{Norm of $\tGamma_1$}
\label{sec:tGamma1Norm}

\begin{lemma}\label{lem:normGamma1}
Let $\tGamma$ be like in~\refeqn{tGamma} with $\tG_T$ defined like in~\refeqn{GT}, and map $\tGamma \deltamaps \tGamma_1$ and $\tG_T \deltamaps (\tG_T)_1$ using \refeqn{E1} and \refeqn{F1}. Then $\|\tGamma_1\| = O(\max_m (\max(\alpha_m m^{(k-1)/2}$,
 $(\alpha_m - \alpha_{m+1}) n^{k/2})))$.
\end{lemma}

\pfstart
\renewcommand{\S}{\mathcal{S}}
We have $\|\tGamma_1\|^2 = \|\tGamma_1^* \tGamma_1\| = \|\sum_S (\tG_S)_1^* (\tG_S)_1\|$.
Decompose the set of all possible $k$-tuples of indices into $\S_1 \cup \S_2$, where $\S_1$ are $k$-tuples containing 1 and $\S_2$ are $k$-tuples that don't contain 1.  We upper-bound the contribution of $\S_1$ to $\|\tGamma_1\|^2$ by $\max_m \alpha_m^2 \binom {m+k-1} {k-1}$ and the contribution of $\S_2$ by $\max_m(\alpha_m - \alpha_{m+1})^2 k \binom {n-1} k$, and apply the triangle inequality.

Let $v = e_{v_1} \otimes \cdots \otimes e_{v_n}$ with $|v| = m+k-1$, and let $S \in \S_1$.  Then, by~\refeqn{F1},
\[
(\tG_S)_1 v = \left\{ \begin{tabular}{l l}
$\alpha_m v^{(1)}$, & $v_1=0$ and $|v_S|=k-1$, \\
0, & otherwise.
\end{tabular} \right.
\]
Here $v^{(1)}= e_{v_2} \otimes \cdots \otimes e_{v_n}$ is $v$ with the first term removed and $v_S = \bigotimes_{s\in S} e_{v_s}$. 

For different $v$, these are orthogonal vectors, and hence $v$ is an eigenvector of $(\tG_S)_1^* (\tG_S)_1$ of eigenvalue $\alpha_m^2$ if $v_1=0$ and $|v_S|=k-1$, and of eigenvalue 0 otherwise. For every $v$ with $v_1=0$ and $|v|=m+k-1$, there are $\binom {m+k-1} {k-1}$ sets $S\in \S_1$ such that $(\tG_S)_1 v \ne 0$. Thus, the contribution of $\S_1$ is as claimed.

Now consider an $S\in \S_2$, that means $1 \not\in S$.
\begin{eqnarray*}
\tG_S &=& \sum_{m=0}^{n-k} \alpha_m F_S \otimes E^{(n-k)}_m \\
&=& \sum_{m=0}^{n-k} \alpha_m F_S \otimes (E_0 \otimes E^{(n-k-1)}_m + E_1 \otimes E^{(n-k-1)}_{m-1}) \\
&\deltamaps& \sum_{m=0}^{n-k} \alpha_m F_S \otimes E_0 \otimes (E^{(n-k-1)}_m - E^{(n-k-1)}_{m-1}) \\
= (\tG_S)_1 &=& \sum_{m=0}^{n-k} (\alpha_m - \alpha_{m+1}) F_S \otimes E_0 \otimes E^{(n-k-1)}_m \enspace.
\end{eqnarray*}
Therefore $(\tG_S)_1$ is of the same form as $\tG_S$, but with coefficients $(\alpha_m - \alpha_{m+1})$ instead of $\alpha_m$ and on one dimension less. We get the required estimate from \reflem{normGamma}(b).

Since $k = O(1)$, we get the claimed bound.
\pfend

\subsection{Optimization of $\alpha_m$}
\label{sec:alpha}

To maximize the adversary bound, we maximize $\|\tGamma\|$ while keeping $\|\tGamma_1\|=O(1)$.  That means, we choose the coefficients $(\alpha_m)$ to maximize $\alpha_0 n^{k/2}$ (\reflem{normGamma}) so that, for every $m$, $\alpha_m \le {m^{(1-k)/2}}$ and $\alpha_m \le \alpha_{m+1} + {n^{-k/2}}$ (\reflem{normGamma1}).


For every $r \in [n]$, $\alpha_0 \le \alpha_r + r {n^{-k/2}} \le {r^{(1-k)/2}} + r {n^{-k/2}}$.
The expression on the right-hand side achieves its minimum, up to a constant, $\alpha_0 = 2\ n^{{k (1-k)}/ (2 (k+1))}$ for $r=n^{k/(k+1)}$.  This corresponds to the following solution:
\begin{equation}
\label{eqn:alphas}
\alpha_m = \max\left\{2 - \frac m {n^{k/(k+1)}}, 0\right\} n^{\frac {k (1-k)} {2 (k+1)}}
\end{equation}
With this choice of $\alpha_m$, $\|\tGamma\| = \Omega(\alpha_0 n^{k/2}) = \Omega(n^{k/ (k+1)})$.

\subsection{Constructing $\Gamma$ from $\tGamma$}
\label{sec:Gamma}

The matrix $\tGamma$ gives us the desired ratio of norms of $\tGamma$ and $\tGamma\circ\Delta_i$. Unfortunately, $\tGamma$ cannot directly be used as an adversary matrix, because it contains illegal columns $y$ with $f(y)=1$, that is, $y$ that contain an element of the orthogonal array on $S \subset [n]: |S|=k$, i.e., $y_S \in T_S$. We show that after removing the illegal columns it is still good enough.

\begin{lemma}
Let $\Gamma$ be the sub-matrix of $\tGamma$ with the illegal columns removed.  Then $\|\Gamma \circ \Delta_1\| \le \|\tGamma \circ \Delta_1\|$, and $\|\Gamma\|$ is still $\Omega(\alpha_0 n^{k/2})$ when $q \ge n^k$.
\end{lemma}

\pfstart
We estimate $\|\Gamma\|$ from below by $w^* \Gamma w'$ using unit vectors $w, w'$ with all elements equal.  Recall Equation \refeqn{GT2}:
\[
\tG_T = \alpha_0 e_0^{\otimes(n-1)}(e_0^{\otimes n})^* + {\sum_{u,v}} \alpha_{|v|} (u^{(\ell_u)}\otimes v)(u\otimes v)^* \enspace,
\]
where the summation is over all $u$ and $v$ such that at least one of them contains an element different from $e_0$. The sum of each column in each of  $(u^{(\ell_u)}\otimes v)(u\otimes v)^*$ still is zero because at least one of $u^{(\ell_u)}$ or $v$ sums up to zero.  Therefore the contribution of the sum is zero regardless of which columns have been removed.

By summing over all $\binom n k$ choices of $S$, we get
\[
\|\Gamma\| \ge w^* \Gamma w' = \sqrt{\binom n k} \alpha_0\ (e_0^{\otimes n})_L^* w' \enspace,
\]
where $e_L$ denotes the sub-vector of $e$ restricted to $L$, and $L$ is the set of legal columns.  Since both $e_0$ and $w'$ are unit vectors with all elements equal, and $w'$ is supported on $L$, $(e_0^{\otimes n})_L^* w' = \sqrt{|L|/q^n}$.

Let us estimate the fraction of legal columns.  The probability that a uniformly random input $y \in [q]^n$ contains an orthogonal array at any given $k$-tuple $S$ is $\frac 1 q$.  By the union bound, the probability that there exists such $S$ is at most $\binom n k \frac 1 q$.  Therefore the probability that a random column is legal is $\frac {|L|} {q^n} \ge 1 - \binom n k \frac 1 q$, which is $\Omega(1)$ when $q \ge n^k$.
\pfend

Thus, with the choice of $(\alpha_m)$ from~\refeqn{alphas}, we have $\Adv(f) = \Omega(\alpha_0 n^{k/2}) = \Omega(n^{k/(k+1)})$. This finishes the proof of \refthm{orthogonal}.

\section{Open problems}


Our technique relies crucially on the $n^k$ lower bound on the alphabet size.  Can one relax this bound in some special cases?  For example, element distinctness is nontrivial when $q \ge n$, but our lower bound only holds for $q \ge n^2$.

A tight $\Omega(n^{2/3})$ lower bound for element distinctness was originally proved by the polynomial method \cite{shi:collisionLower} by reduction via the collision problem.  The $k$-collision problem is to decide whether a given function is $1:1$ or $k:1$, provided that it is of one of the two types.  One can use an algorithm for element distinctness to solve the 2-collision problem, and thus the tight $\Omega(n^{1/3})$ lower bound for collision in \cite{shi:collisionLower} implies a tight lower bound for element distinctness.  Unfortunately, the reduction doesn't go in both directions and hence our result doesn't imply any nontrivial adversary bound for $k$-collision.  The simpler non-negative adversary bound is limited to $O(1)$ due to the property testing barrier.  Roughly speaking, if every 0-input differs from every 1-input in at least an $\varepsilon$-fraction of the input, the non-negative adversary bound is limited by $O(\frac 1 \varepsilon)$.  How does an explicit negative adversary matrix for an $\omega(1)$ lower bound look like?

The recent learning graph-based algorithm for $k$-distinctness \cite{belovs:learningKDist} uses $O(n^{1 - 2^{k-2} / (2^k-1)})$ quantum queries, which is less than $O(n^{k/(k+1)})$ but more than the $\Omega(n^{2/3})$ lower bound by reduction from 2-distinctness.  $k$-distinctness is easier than the $k$-sum problem considered in our paper because one can obtain nontrivial information about the solution from partial solutions, i.e., from $\ell$-tuples of equal numbers for $\ell < k$.  Can one use our technique to prove an $\omega(n^{2/3})$ lower bound for $k$-distinctness?

The $k$-sum problem is very structured in the sense that all $k$-tuples of the input variables, and all possible values seen on a $(k-1)$-tuple, are equal with respect to the function.  The symmetry of this problem helped us to design a symmetric adversary matrix.  The nonnegative adversary bound gives nontrivial lower bounds for most problems, by simply putting most of the weight on hard-to-distinguish input pairs, regardless of whether the problem is symmetric or not.  Can one use our technique to improve the best known lower bounds for some non-symmetric problems, for example, to prove an $\omega(\sqrt n)$ lower bound for graph collision, $\omega(n)$ for triangle finding, or $\omega(n^{3/2})$ for verification of matrix products?

\subsection*{Acknowledgments}
A.B. would like to thank Andris Ambainis, Troy Lee and Ansis Rosmanis for valuable discussions.  We are grateful to Kassem Kalach for informing about the applications of the $k$-sum problem in Merkle puzzles, and for reporting on some minor errors in the early version of the paper.

A.B. is supported by the European Social Fund within the project ``Support for Doctoral Studies at University of Latvia'' and by FET-Open project QCS.


\bibliographystyle{../../habbrvE}
\bibliography{../../bib}

\begin{thebibliography}{10}

\bibitem{shi:collisionLower}
S.~Aaronson and Y.~Shi.
\newblock Quantum lower bounds for the collision and the element distinctness
  problems.
\newblock {\em Journal of the ACM}, 51(4):595--605, 2004.

\bibitem{ambainis:adversary}
A.~Ambainis.
\newblock Quantum lower bounds by quantum arguments.
\newblock {\em Journal of Computer and System Sciences}, 64(4):750--767, 2002,
  \href{http://xxx.lanl.gov/abs/quant-ph/0002066}{{\ttfamily
  arXiv:quant-ph/0002066}}.

\bibitem{ambainis:polynomialVsQCC}
A.~Ambainis.
\newblock Polynomial degree vs. quantum query complexity.
\newblock In {\em Proc. of 44th IEEE FOCS}, pages 230--239, 2003,
  \href{http://xxx.lanl.gov/abs/quant-ph/0305028}{{\ttfamily
  arXiv:quant-ph/0305028}}.

\bibitem{ambainis:collisionLower}
A.~Ambainis.
\newblock Polynomial degree and lower bounds in quantum complexity: Collision
  and element distinctness with small range.
\newblock {\em Theory of Computing}, 1:37--46, 2005,
  \href{http://xxx.lanl.gov/abs/quant-ph/0305179}{{\ttfamily
  arXiv:quant-ph/0305179}}.

\bibitem{ambainis:distinctness}
A.~Ambainis.
\newblock Quantum walk algorithm for element distinctness.
\newblock {\em SIAM Journal on Computing}, 37:210--239, 2007,
  \href{http://xxx.lanl.gov/abs/quant-ph/0311001}{{\ttfamily
  arXiv:quant-ph/0311001}}.

\bibitem{beals:polynomial}
R.~Beals, H.~Buhrman, R.~Cleve, M.~Mosca, and R.~de~Wolf.
\newblock Quantum lower bounds by polynomials.
\newblock {\em Journal of the ACM}, 48(4):778--797, 2001,
  \href{http://xxx.lanl.gov/abs/quant-ph/9802049}{{\ttfamily
  arXiv:quant-ph/9802049}}.

\bibitem{belovs:adv-el-dist}
A.~Belovs.
\newblock Adversary lower bound for element distinctness.
\newblock 2012,  \href{http://xxx.lanl.gov/abs/1204.5074}{{\ttfamily
  arXiv:1204.5074}}.

\bibitem{belovs:learningKDist}
A.~Belovs.
\newblock Learning-graph-based quantum algorithm for $k$-distinctness.
\newblock 2012,  \href{http://xxx.lanl.gov/abs/1205.1534}{{\ttfamily
  arXiv:1205.1534}}.

\bibitem{brassard:merkle}
G.~Brassard, P.~H{\o}yer, K.~Kalach, M.~Kaplan, S.~Laplante, and L.~Salvail.
\newblock Merkle puzzles in a quantum world.
\newblock In {\em CRYPTO 2011}, pages 391--410. Springer, 2011,
  \href{http://xxx.lanl.gov/abs/1108.2316}{{\ttfamily arXiv:1108.2316}}.

\bibitem{buhrman:querySurvey}
H.~Buhrman and R.~de~Wolf.
\newblock Complexity measures and decision tree complexity: a survey.
\newblock {\em Theoretical Computer Science}, 288:21--43, 2002.

\bibitem{childs:subsetFinding}
A.~Childs and J.~Eisenberg.
\newblock Quantum algorithms for subset finding.
\newblock {\em Quantum Information \& Computation}, 5(7):593--604, 2005,
  \href{http://xxx.lanl.gov/abs/quant-ph/0311038}{{\ttfamily
  arXiv:quant-ph/0311038}}.

\bibitem{hoyer:adversaryNegative}
P.~H{\o}yer, T.~Lee, and R.~{\v S}palek.
\newblock Negative weights make adversaries stronger.
\newblock In {\em Proc. of 39th ACM STOC}, pages 526--535, 2007,
  \href{http://xxx.lanl.gov/abs/quant-ph/0611054}{{\ttfamily
  arXiv:quant-ph/0611054}}.

\bibitem{kalach:personal}
K.~Kalach.
\newblock Personal communication, 2012.

\bibitem{kutin:collisionLower}
S.~Kutin.
\newblock Quantum lower bound for the collision problem with small range.
\newblock {\em Theory of Computing}, 1(1):29--36, 2005.

\bibitem{lee:stateConversion}
T.~Lee, R.~Mittal, B.~Reichardt, R.~{\v{S}}palek, and M.~Szegedy.
\newblock Quantum query complexity of the state conversion problem.
\newblock In {\em Proc. of 52nd IEEE FOCS}, pages 344--353, 2011,
  \href{http://xxx.lanl.gov/abs/1011.3020}{{\ttfamily arXiv:1011.3020}}.

\bibitem{reichardt:adversaryTight}
B.~Reichardt.
\newblock Reflections for quantum query algorithms.
\newblock In {\em Proc. of 22nd ACM-SIAM SODA}, pages 560--569, 2011,
  \href{http://xxx.lanl.gov/abs/1005.1601}{{\ttfamily arXiv:1005.1601}}.

\bibitem{spalek:adversaryEquivalent}
R.~{\v S}palek and M.~Szegedy.
\newblock All quantum adversary methods are equivalent.
\newblock {\em Theory of Computing}, 2:1--18, 2006,
  \href{http://xxx.lanl.gov/abs/quant-ph/0409116}{{\ttfamily
  arXiv:quant-ph/0409116}}.

\bibitem{zhang:adversaryPower}
S.~Zhang.
\newblock On the power of ambainis lower bounds.
\newblock {\em Theoretical Computer Science}, 339(2):241--256, 2005,
  \href{http://xxx.lanl.gov/abs/quant-ph/0311060}{{\ttfamily
  arXiv:quant-ph/0311060}}.

\end{thebibliography}

\end{document}